\documentclass[conference]{IEEEtran}

\ifCLASSINFOpdf
 
\else
  
\fi

\usepackage{graphicx}
\usepackage{amsfonts}
\usepackage{amsmath}
\usepackage{amsthm}
\usepackage{amssymb}
\usepackage{ulem} 
\newtheorem{lemma}{Lemma}
\newtheorem{theorem}{Theorem}
\newtheorem{definition}{Definition}
\newtheorem{remark}{Remark}
\newtheorem{corollary}{Corollary}
\newtheorem{folk}{Folk Lemma}
\usepackage{xcolor}
\usepackage[left=1.7cm,right=1.7cm,top=0.73 in,bottom= 1.1in]{geometry}

\newcommand{\F}{\mathbb{F}}

\newcommand{\N}{\mathbb{N}}
\newcommand{\supp}{\textnormal{supp}}
\newcommand{\rmv}[1]{}

\begin{document}

\title{Algebraic Geometric Rook Codes for Coded Distributed Computing}

\author{\IEEEauthorblockN{Gretchen L. Matthews \& Pedro Soto}
\IEEEauthorblockA{Department of Mathematics\\Virginia Tech\\Blacksburg, Virginia 24061 USA\\
\{gmatthews, pedrosoto\}@vt.edu}
}

\maketitle

\begin{abstract}
We extend coded distributed computing over finite fields to allow the number of workers to be larger than the field size.
We give codes that work for fully general matrix multiplication and show that in this case we serendipitously have that all functions can be computed in a distributed fault-tolerant fashion over finite fields.
This generalizes previous results on the topic. We prove that the associated codes achieve a recovery threshold similar to the ones for characteristic zero fields but now with a factor that is proportional to the genus of the underlying function field.
In particular, we have that the recovery threshold of these codes is proportional to the classical complexity of matrix multiplication by a factor of at most the genus.
\end{abstract}

\IEEEpeerreviewmaketitle

\section{Introduction}
In this paper we consider the problem of coded distributed computation over a finite field.  
Coded distributed computing and, in particular, coded distributed matrix multiplication has attracted a large surge of research interest as of late~\cite{LLPPR2018, yma17, dfhjcg2020, YAA2018, DBJMG2018, tldk17,sigl22, sl20, sfsl22, ylrksa19, jj21, fc21, sma21, rk20, cgw21,jj21a, okk02024,schwartz2024}.
In this paper we will extend the batch matrix multiplication problem in \cite{ylrksa19, jj21,sl20, sfsl22,  censorhillel2023nearoptimal} to the case where there are more workers than there are elements in the field. 
We show that over finite fields, our rook codes can encode \emph{all functions.} We use codes constructed from algebraic function fields. Prior works that use algebraic geometry codes include \cite{fidalgodíaz2024distributed},  \cite{Okko}, \cite{HerA}, and evaluation codes \cite{SecureMatDot}.  

This paper is organized as follows. Section~\ref{sec:diag_rook_mat} gives an implicit construction that solves the general coded distributed matrix-matrix multiplication problem that is optimal to a factor of $2(g+1)$ where $g$ is the genus of a particular function field,  
Section~\ref{sec:ent_rook_mat} gives an explicit construction that is good for small values, 
and Section~\ref{sec:tensors} gives a construction that computes \emph{any function} which is optimal up to a factor of $\ell(g+1)$ where $\ell$ is the degree of the function and $g$ is also the genus of a yet unspecified function field.

\paragraph*{Background and Notation}
Consider a function field $F$ of genus $g$ over a finite field $\F$. The set of places of $F$ is denoted by $\mathbb P_F$.
 The divisor of a nonzero rational function $f \in F$ is 
$(f)=(f)_0-(f)_{\infty}$ where 
$(f)_0:=\sum_{P \in \mathbb P_F, v_P(f) > 0} v_P(f) P$ and $(f)_{\infty}:=\sum_{P \in \mathbb P_F, v_P(f) > 0} v_P(f) P $ denote the zero and pole divisors of $f$ and $v_P(f)$ denotes the discrete valuation of $f$ at the place $P$. 

Consider a divisor $G=\sum_{P \in \mathbb P_F} a_P P$ of $F$. Its degree is $\deg(G):=\sum_{P \in \mathbb P_F} a_P \deg (P)$, and its support is $\supp(G):=\left\{ P \in \mathbb P_F: a_P \neq 0 \right\}$.
The Riemann-Roch space of  $G$ is 
$\mathcal L(G):= \left\{ f : (f) \geq -G \right\} \cup \{ 0 \}$, meaning $f \in \mathcal L(G)$ if and only if $f$ has a zero of order at least $-a_P$ for each $P \in \supp (G)$ with $a_P<0$ and the only poles of $f$ are at 
$P \in \supp (G)$ with $a_P>0$ and of pole order at most $a_P$. The dimension of the divisor $G$ is $\ell(G):=\dim_{\F}\mathcal L(G)$. If $\deg(G) > 2g-2$, then according to the Riemann-Roch Theorem, $\ell(G)=\deg G + 1 - g$. 

Given divisors $G$ and $D:=P_1+\dots+P_n$ on $F$ with disjoint support where each $P_i$ is a rational place, the associated algebraic geometry code is  $C(D,G):=\left\{ \left( f(P_1), \dots, f(P_n) \right): f \in \mathcal L(G) \right\}$. It is well known that $C(D,G)$ is an $[n,k,d]$ code over $\F$, with length $n$, dimension $k=\ell(G)-\ell(G-D)$ and minimum distance $d \geq n-\deg G$. Hence, if $2g-2< \deg(G)<n$, then $k=\deg(G)+1-g$. For additional details, see \cite{Stichtenoth}.

\section{Diagonal Algebraic Geometric Rook Product}\label{sec:diag_rook_mat}

\subsection{Batch Matrix Multiplication}
We will consider the following problem: given $k $ pairs of matrices 
$$
A_1,B_1,\dots,A_k,B_k,
$$
where  $A_i \in \F^{t_1 \times t_{2}}$ and $B_i \in \F^{t_2 \times t_{3}}$ for $i \in [k]$ and a field $\F$, 
compute the products 
$$
A_1 B_1 ,\dots,A_k B_k
$$
in the distributed master worker topology in which the master node  gives  $n$ worker nodes  coded matrices of the form 
$$
\tilde A _w = \sum_{i \in [k]} \alpha_{i}^{(w)} A_i  
\in \F^{t_1 \times t_{2}}
, \tilde B _w = \sum_{i \in [k]} \beta_{i}^{(w)} B_i
 \in \F^{t_2 \times t_{3}}
$$
where $\alpha_{i}^{(w)} \in \F$ and $w \in [n]$ indexes over the worker nodes. 
We are concerned with the minimum number of worker nodes that need to return their values so that the master can recover the desired products. This will be formalized in Definition~\ref{def:rec_bat}.
Before we move on to constructing the actual codes, we will show that this batch matrix problem is actually the most general form of the distributed matrix multiplication problem since it implicitly solves the \textbf{general matrix-matrix} multiplication problem.
\subsection{General Matrix-Matrix Multiplication}\label{sec:mat_mult_explaination}

Given two matrices
\begin{equation*}
    A = \begin{bmatrix}
        A_{1,1} & \dots & A_{1,\zeta}\\
         \vdots & \ddots & \vdots \\
        A_{\chi,1} & \dots & A_{\chi,\zeta}\\
    \end{bmatrix} 
    , \ \ 
    B = \begin{bmatrix}
        B_{1,1} & \dots & B_{1,\upsilon}\\
         \vdots & \ddots & \vdots \\
        B_{\zeta,1} & \dots & B_{\zeta,\upsilon}\\
    \end{bmatrix}, 
\end{equation*} with
$A_{i,j} \in \F^{t_1 \times t_2}$ and $B_{i,j} \in \F^{t_2 \times t_3}$, we can take an optimal $(\chi, \upsilon, \zeta)$ fast matrix multiplication tensor of rank $r$, {i.e.,} for $t \in [r]$, take
$$
\hat A _t = \sum_{i,j \in [\chi] \times [\zeta]} \gamma_{i,j}^{(t)} A_{i,j} , \ \ \hat B _t = \sum_{i,j \in [\zeta ]\times [\upsilon ]} \delta_{i,j}^{(t)} B_{i,j},
$$
    and then  $A_i,B_i := \hat{A}_i,\hat{B}_i $ as in \cite{yu_ali_ave_2020}, which shows that the general matrix-matrix multiplication problem reduces exactly to the batch matrix multiplication. It seems that computing the linear functions defined by $\gamma,\delta$ is an undesirable overhead, \textit{but one must compute the functions given by $\alpha,\beta$ anyways; in particular, one may compose the $\alpha,\beta$ and the $\gamma, \delta$ so that there is no overhead.} The number $r$ above is often called the \textit{tensor-rank} or the \textit{bilinear complexity}.

The primary goal of this paper is to generalize the main result of \cite{yu_ali_ave_2020} (which only holds in true generality generally for characteristic zero) to more general cases over finite fields. 
This paper 
overcomes the major obstacle when the field is finite, namely the number of evaluation points (meaning the number of workers) or more generally, the length of the code.  

\subsection{Implicit Construction}

Let $F/\mathbb{F}_q$ be a function field of transcendence degree 1.
\begin{definition}\label{def:rec_bat}
Given a divisor $G$ on $F$, we say that $\mathcal L(G)$ is 
$R$-recoverable for a positive integer $R$ if there exist functions $x_1,\dots,x_k \in \mathcal{L}(G_1)$, $y_1,\dots,y_k \in \mathcal{L}(G_2)$ 
for some divisors $G_1$ and $G_2$ of $F$ with $G=G_1+G_2$ 
so that 
the values of $A_i B_i$ (i.e., the diagonal elements) can be recovered from any $R$ columns of the product
\begin{equation} \label{eq:diag_code}
   C
    \begin{bmatrix}
        x_1y_1(P_1) & x_1y_1(P_2)  & \dots & x_1y_1(P_n) \\  
 \vdots & \vdots  &  & \vdots \\ 
         x_1y_k(P_1) & x_1y_k(P_2)  & \dots & x_1y_k(P_n) \\ 
 x_2y_1(P_1) & x_2y_1(P_2)  & \dots & x_2y_1(P_n) \\ 
         \vdots & \vdots &  & \vdots& \\ 
         x_2y_k(P_1) & x_2y_k(P_2)  & \dots & x_2y_k(P_n) \\  
x_ky_1(P_1) & x_ky_1(P_2)  & \dots & x_ky_1(P_n) \\ 
         \vdots & \vdots &  & \vdots& \\ 
         x_ky_k(P_1) & x_ky_k(P_2)  & \dots & x_ky_k(P_n) \\  
    \end{bmatrix} 
\end{equation}
for some matrix
\begin{equation}\label{eq:data}
   C = (A_1B_1 \cdots A_1B_k A_2B_1 \cdots A_2B_k \cdots A_kB_1\cdots A_kB_k) 
\end{equation}
where $A_iB_j \in \F^{t_1 \times t_3}$ for all $i,j \in [k]$.  In this case, we may say that $G$ is $R$-recoverable.
\end{definition}

\begin{definition}\label{def:diag_rook_prod}
 Let $G$ and $D = P_1+\dots+P_n$ be divisors with disjoint supports on $F / \mathbb{F}_q(x)$. We call $\mathcal{C}(G,D)$ a (diagonal) $[n,k]_q-$\textbf{rook code} if there exist  bases $\{x_1,\dots,x_k\}$ for $\mathcal{L}(G_1)$ and $\{y_1,\dots,y_k\}$ for $\mathcal{L}(G_2)$ 
satisfying for all $i,j \in [k]$ and $l \in [n]$\begin{equation}\label{eq:nice_resi}
        P_l \not\in \mathrm{supp}(x_iy_j)  \iff i = j = l
    \end{equation} 
    where 
 $G_1$ and $G_2$ are divisors on $F / \mathbb{F}_q(x)$ such that $G=G_1+G_2$.
\end{definition}

\begin{lemma}
    \label{lem:alt_def}
    An $[n,k]_q$ rook code given by bases 
    $\{x_1,\dots,x_k\}$  and $\{y_1,\dots,y_k\}$ 
    satisfies the following generalization of the decodability condition of \cite{censorhillel2023nearoptimal}  for all $i,j,k$
    \begin{equation*}
        (x_k) + (y_k) = (x_i) + (y_j) \iff k = i \land  k = j.
    \end{equation*}
\end{lemma}

\begin{proof}
    The proof is immediate from the fact that $(x_iy_j) = (x_i) + (y_j)$. Thus, the two expressions have the same supports. 
\end{proof}

\begin{remark}
The term rook code is inspired by the name given to the codes in \cite{censorhillel2023nearoptimal}, since Lemma~\ref{lem:alt_def} shows that Definition~\ref{def:diag_rook_prod} is a generalization of the decodability condition presented in \cite{censorhillel2023nearoptimal}. 
    They are equivalent up to code equivalance.
\end{remark}

\begin{lemma}
    If $x_1,\dots,x_k \in \mathcal{L}(G_1)$ and $y_1,\dots,y_k \in \mathcal{L}(G_2)$ satisfy 
Equation~\ref{eq:nice_resi} and $\text{supp}(G_1+G_2) \cap \text{supp}(D) = \emptyset$,
    then $G_1+G_2$ is $R$-recoverable for some $ R \leq d(G)$. 
\end{lemma}

\begin{proof}
Worker $w$ will receive the values 
\begin{equation*}
    \tilde{A}(P_w) = \sum_{i \in [k]} A_ix_i(P_w), \ \ \tilde{B}(P_w) = \sum_{i \in [k]} B_ix_i(P_w).
\end{equation*}
The matrix consisting of their products will be exactly the matrix given by Equation~\ref{eq:diag_code}. 
In particular, worker $w$ will return  $$\tilde{C}(P_w)=\tilde{A}(P_w) \tilde {B}(P_w)$$ to the master node.
    The proof now follows from the fact that $n- \deg  (G) \leq  d$, where $d$ denotes the minimum distance of $C(D,G)$, is the maximum number of rational places needed to recover the $k$ blocks of data $A_iB_i$ which is a subset of the indices in Equation~\ref{eq:data}.
\end{proof}
\subsection{Existence of Codes for All Fields and Numbers of Rational Places }

Next, we demonstrate that diagonal rook codes exist over every field. 

\begin{theorem}\label{thm:exi}
    Given a function field  $F/\F_q$ with at least $n$ rational places, there exists an $[n,k]_q$-rook code for any $k \in [n]$.
\end{theorem}
\begin{proof}
    Fix an element $x \in F$ that is transcendental over $\mathbb{F}_q$ and rational places $P_1, \dots, P_k$ of $F$. 
    By repeated application of the Approximation Theorem (see \cite[Theorem 1.3.1]{Stichtenoth} for instance), there exists some $a_1,\dots,a_k \in F$ such that
    \begin{equation*}
        v_{P_i}(x-a_i) = 0
    \end{equation*}
    and 
    \begin{equation*}
        v_{P_j}(x-a_i) = 1
    \end{equation*}
    for $i\neq j$. Then the functions
    \begin{equation*}
        x_i = y_i = \prod_{j \in [k]\setminus \{i\}} x - a_i. 
    \end{equation*}
    satisfy Equation~\ref{eq:nice_resi} by construction. 
\end{proof}

\begin{remark}
The proof of Theorem~\ref{thm:exi} suggests another possible generalization of the decodability condition of \cite{censorhillel2023nearoptimal}, namely, 
    \begin{equation*}
       v_{P_k} (x_k) + v_{P_k}(y_k) = v_{P_k}(x_i) + v_{P_k}(y_j) \iff k = i \land  k = j.
    \end{equation*}
\end{remark}

\subsection{A More Efficient Construction }

For $i \in [k]$, let $r_i:=\min \left\{ \alpha \in  H(P_i): \alpha > 0 \right\}$ where $H(P)=\left\{ \alpha \in \N: \ell(\alpha P) \neq \ell ((\alpha -1)P) \right\}$ is the Weierstrass semigroup of a rational place $P$. Then there exist functions $z_i$ such that 
\begin{equation*} \label{E:z}
    \left\langle 1, z_i \right\rangle = \mathcal{L} (r_iP_i).
\end{equation*}
Then we define 
\begin{equation}\label{eq:main_rook}
    x_i := \prod_{ j\in [k]\setminus \{i\}} z^{-1}_j.
\end{equation}
Our coding scheme will send matrices 
\begin{equation*}
    \tilde{A}(P_w) = \sum_{i \in [k]} A_ix_i(P_w), \ \ \tilde{B}(P_w) = \sum_{i \in [k]} B_ix_i(P_w)
\end{equation*}
to worker $w$. Let $$b = \sum_{i \in [k]}r_i\mathrm{deg}(P_i).$$ Then we have that
\begin{equation*}
    \left\langle 1, x_1,\dots,x_k \right\rangle \subset \mathcal{L} (Q_1+\dots+Q_k),
\end{equation*}
where $Q_i = (z_i^{-1})_\infty$ and $\mathrm{deg}(\sum_i Q_i) \leq b$. Assume further that there exist rational places 
$P_{k+1}, \dots, P_{n+k}$ of $F$
so that 
$$
\mathrm{supp}(D) \cap \mathrm{supp}(G) = \emptyset
$$
where $G=Q_1+\dots+Q_b$ and $D = P_{k+1}+\dots+P_{n+k}$. Then consider the code $\mathcal{C}(D,G)$.

Before proving that the previous construction is indeed a rook code we introduce the following measure of complexity which will turn out to (upper) bound the recovery threshold.
\begin{definition}
    Given a place $P$ of $F / \mathbb{F}_q$, the \textbf{min pole number} of $P$, denoted $\mu(P)$ as the smallest integer $r $ such that there exists a non-constant $z \in F$ such that 
    \begin{equation*}
    \left\langle 1, z \right\rangle = \mathcal{L} (rP);
    \end{equation*}
    that is, the min pole number of $P$ is the multiplicity of the Weierstrass semigroup of the place $P$.
    We define the ($k^\mathrm{th}$) \textbf{min pole sum} as 
    \begin{equation*}
        \sigma_k(F) : = \min \left\{ \sum_{i\in[k]}\mu(P_i): \begin{array}{l} P_i \in \mathbb{P}_F,  \deg(P_i)=1,\\ P_i \neq P_j  \ \forall i, j \in [k] , \\ \textnormal{ with } i \neq j \end{array} \right\}.
    \end{equation*}
\end{definition}

\begin{theorem}\label{thm:main_result}
 The functions $x_i$ given in  (\ref{eq:main_rook}) satisfy Equation~\ref{eq:nice_resi}.  The associated construction   has recovery threshold given by 
 \begin{equation}\label{eq:main_result}
    R = 2\sigma_k(F). 
 \end{equation}
\end{theorem}
\begin{proof}
Since the $P_i$ were chosen to be distinct and $\mathrm{supp}(z_i^{-1})_0 = \mathrm{supp}(z_i)_\infty = \{P_i\}$, we have that 
\begin{equation*}
  \bigcup_{j \in [k] \setminus \{i\}}   \mathrm{supp}(z_j)_\infty \\ 
   = \bigcup_{j \in [k] \setminus \{i\}} \mathrm{supp}(z^{-1}_j)_0 = \mathrm{supp}(x_i)_0
\end{equation*}
We have that
\begin{equation*}
    P_i \notin  \mathrm{supp}(z_j)_ \infty ,
\end{equation*}
for $i \neq  j$ by construction. Therefore,
\begin{equation*}
   P_i \notin  \bigcup_{j \in [k] \setminus \{i\}} \mathrm{supp}(z_j)_\infty \rmv{
   = \bigcup_{j \in [k] \setminus \{i\}} \mathrm{supp}(z^{-1}_j)_0} = \mathrm{supp}(x_i)_0.   
\end{equation*}
Similarly, we have that 
$
    P_\ell \in  \mathrm{supp}(z_j)_ \infty
$
for $ \ell = j $ by construction; therefore,
\begin{equation*}
   P_\ell \in \bigcup_{j \in [k] \setminus \{i\}}\mathrm{supp}(z_j)_\infty \rmv{= \bigcup_{j \in [k] \setminus \{i\}} \mathrm{supp}(z^{-1}_j)_0 }= \supp (x_i)_0,   
\end{equation*}
for any $\ell \in [k]$ such that $\ell \neq i $, by construction.
Therefore,  Equation~\ref{eq:nice_resi} is satisfied.
We then have that $$b = \sum_{i \in [k]}r_i\mathrm{deg}(P_i) = \sigma_k(F),$$
for an optimal choice of $P_1,\dots,P_k$.

\end{proof}

\subsection{Analysis: Upper Bounds on the Recovery Threshold}
The bound in Theorem~\ref{thm:main_result} can be given a coarse upper bound as stated in the next result.
\begin{theorem}
   The codes defined by Equation~\ref{eq:main_rook} have the following upper bound:
   \begin{equation*}
       R(n,k,q) \leq \left( g_{n,q}+1 \right) k
   \end{equation*}
   where $g_{n,q}$ is the smallest genus of a function field over $\mathbb{F}_q$ with at least $(g_{n,k}+1)k+n$ rational places.
\end{theorem}
\begin{proof}
   According to the Weierstrass Gap Theorem, for any rational place $P$ of a function field $F$ of genus $g$, $\N \setminus H(P) \subseteq [1, \dots, 2g-1]$ and $\mid \N \setminus H(P) \mid =g$. Hence, $\mu(P) \leq g +1$ and $\sigma_k(F) \leq (g+1)k$ provided $\F$ has at least $k$ rational places. The result then follows from the fact that the scheme requires $n$ reserved rational places.
\end{proof}
\rmv{\subsubsection{Hasse-Weil and Lower Bounds on the Recovery }
Prove that what Umberto and his student tried to do won't work. What goes wrong is a generalization of \cite{censorhillel2023nearoptimal}}

\begin{corollary}
    The recovery threshold $R$ of a scheme that computes the corresponding general matrix-matrix coded distributed version satisfies the bound
    \begin{equation*}
      \mathcal{T} \leq  R \leq  2\sigma_{\mathcal{T}} (F) \leq \left( g_{n,q}+1 \right) \mathcal{T}.
    \end{equation*}
    In particular, if we take $\chi = \zeta = \upsilon = \tau$, then we have that the recovery threshold asymptotically bounded as
    \begin{equation*}
        R   = O(\sigma_{\tau^{\omega}}(F)) = O(g_{n,q}\tau^{\omega}), 
    \end{equation*}
    where $\omega$ is the matrix multiplication exponent.
\end{corollary}

\begin{proof}
    {Repeating the arguements in \cite{yu_ali_ave_2020} and Section~\ref{sec:mat_mult_explaination}, the recovery threshold for batch matrix multiplication bounds the complexity of general matrix multiplication by a factor of two.}
\end{proof}

\begin{remark}
    It is possible that the recovery threshold is far smaller than the genus. In particular, for hyperelliptic function fields we have that $\sigma_k(F) = 2k$.
    Future research would involve finding other families of curves where we can replace the genus bound on recovery threshold with the \textbf{gonality.} 
\end{remark}

\section{Entangled Algebraic Geometric Rook Product}\label{sec:ent_rook_mat}
In order to separate the different constructions, we will call the rook codes from the previous section \textit{diagonal rook codes} (diagonal codes for short) and the code from this section \textit{entangled rook codes} (entangled codes for short). However, we will see in Section~\ref{sec:tensors} that diagonal rook codes can code the most general functions in a straightforward way when the field is finite. 
We will see that the difference between the two is that the diagonal codes implicitly encode matrix multiplication while the entangled codes attempt to do two things at once: 1) code the matrices and 2) be a fast matrix multiplication tensor. 
Diagonal codes, in contrast, simply take an already optimal fast matrix multiplication tensor and encode that as a batch matrix multiplication. 

\subsection{Entangled Codes Do Matrix-Matrix Multiplication Well for Small Cases}
In the classic characteristic zero case, codes of this form achieve 
the naive cubic recovery threshold at best and thus it is unlikely they perform as well as the diagonal ones. 
However, for small values of $k$ they do better. For example, using the entangled polynomial codes construction from \cite{yu_ali_ave_2020}, we get that for $A$ a $2 \times 2$ matrix and $B$ a $2 \times 2 $ matrix, that the entangled polynomial codes have a recovery threshold of $9=2*2*2+2-1$ while the LCC \cite{ylrksa19} and CSA \cite{jj21} constructions \rmv{from CITE} achieves a recovery threshold of $13=2*7-1$. 
For more explanation, please see \cite{yu_ali_ave_2020}.
\subsection{Entangled Codes as an Explicit Construction}
In the entangled rook code case, instead of implicitly giving the general matrix multiplication as a batch of $k = \mathcal{T}$  matrix multiplications, one directly looks for code that also performs fast matrix multiplication simultaneously. 
Since the implicit batch multiplication can already bring the recovery threshold to within a factor of $\sigma_K(F) \leq 2(g+1)$ (or just a factor of $2$ in the case where $\mathbb{F}$ is an infinite field, since one can use MDS-codes without running out of rational place), this entails trying to get the factor down beneath the diagonal rook codes. 

For the entangled codes, we need to redefine what recovery and rook codes means.

\begin{definition}\label{def:rec_ent}
Given $x_{1,1},\dots,x_{\chi, \zeta} \in \mathcal{L}(G_1)$, $y_{1,1},\dots,y_{\zeta, \upsilon} \in \mathcal{L}(G_2)$, and $x_{i,j}y_{k,\ell} \in \mathcal{L}(G_1 + G_2)$, we say that $\mathcal{L}(G_1 + G_2)$ is $R$-recoverable if the values of $\sum_{j \in \zeta }A_{i,j} B_{j,k}$ (i.e., the dot products) can be recovered from any $R$ columns of the result of 
\begin{equation*} 
   C
    \begin{bmatrix}
        x_{1,1}y_{1,1}(P_1) & x_{1,1}y_{1,1}(P_2)  & \dots & x_{1,1}y_{1,1}(P_n) \\  
         x_{1,1}y_{1,2}(P_1) & x_{1,1}y_{1,2}(P_2)  & \dots & x_{1,1}y_{1,2}(P_n) \\  
         \vdots & \vdots & \ddots & \vdots& \\ 
         x_{\chi, \zeta}y_{\zeta, \upsilon}(P_1) &          x_{\chi, \zeta}y_{\zeta, \upsilon}(P_2)  & \dots & x_{\chi, \zeta}y_{\zeta, \upsilon}(P_n) \\  
    \end{bmatrix}, 
\end{equation*}
where 
\begin{equation*}
\resizebox{0.9\hsize}{!}{$
   C =   \begin{bmatrix}
        A_{1,1}B_{1,1} &  A_{1,1}B_{1,2} & \dots & A_{\chi,\zeta}B_{1,1} &\dots & A_{\chi,\zeta}B_{\zeta,\upsilon}
    \end{bmatrix}. 
    $}
\end{equation*}

\end{definition}

\begin{definition}\label{def:ent_rook_prod}
 Let $G = G_1+G_2$ and $D = P_1+\dots+P_n$. We call $\mathcal{C}(G,D)$ a $[n,k]_q-$\textbf{entangled rook code}, if there are bases $x_{i,j}$ for $\mathcal{L}(G_1)$ and $y_{k,\ell}$ for $\mathcal{L}(G_2)$ such that, 
 \rmv{I don't think we need this anymore: after reindexing places with triple indices, we have that }
 \begin{equation}\label{eq:nice_resi_ent}
(x_{i,j}y_{k,\ell}) = (x_{i',j'}y_{k',\ell'}) \iff  j = k = j' = k'.
    \end{equation} is satisfied. 
\end{definition}

Assume there is some $r$ such that 
$
    \left\langle 1, z \right\rangle = \mathcal{L} (rP)
$
for some rational places $P$.
Then we define   
\begin{equation*}
    x_{i,j} := (z^{\upsilon \zeta i + j})^{-1}, \ \ y_{k ,\ell} := (z^{\zeta\ell+ \zeta  - k})^{-1}
\end{equation*}
then we have that
\begin{equation*}
(x_{i,j}y_{k,\ell})_0 = r((\upsilon \zeta i + j+\zeta\ell + \zeta  - k)P_{0}).
\end{equation*}
It should be clear that only when $j=k$ do we have 
\begin{equation*}
(x_{i,j}y_{k,\ell})_0 = r_0(\upsilon\zeta i + \zeta \ell + \zeta  ) P_{0},
\end{equation*}
and thus (after normalizing the $x_{i,j}$ and $y_{k,\ell}$) we have that the coefficient of $x_{i,k}y_{k,\ell}$ in the product $( \sum_{i,j}A_{i,j}x_{i,j})(\sum_{k,\ell}B_{k,\ell}y_{k,\ell})$ is equal to $\sum_{k}A_{i,k}B_{k,\ell}$.
Thus we have an alternate coding scheme that achieves a cubic recovery threshold. 

\begin{remark}
    We postpone the analysis of the previous construction since it would asymptotically give a cubic recovery threshold (\emph{i.e.,} the complexity of naive matrix multiplication) up to a factor proportional to the genus. 
    It is likely that one can extend the impossibility results from \cite{censorhillel2023nearoptimal} that were proven using additive combinatorics to the case of the semigroup of only one point. 
    The main intuition  behind the diagonal design is to consider semigroups of many points, allowing for more elbow room in the construction so that such impossibility results do not hinder us.
\end{remark}
\section{Diagonal Rook Codes for Tensors}\label{sec:tensors}
\subsection{Multi-linear Functions}
By a \textbf{tensor}, $T$, we mean a function 
$
    T: V_1,\dots,V_{\ell} \rightarrow \mathbb{F}
    $
\begin{multline*}
    T(v_1,\dots,\alpha v_i + \beta w, \dots,v_\ell) 
    \\ = \alpha T(v_1,\dots, v_i , \dots,v_\ell) + \beta T(v_1,\dots, w, \dots,v_\ell)
\end{multline*}
for all $i$; we call $\ell$ the order of the tensor. Given bases $\mathcal{B}_i$ for the $V_i$ we can represent a tensor by the values $T$.
\begin{definition}\label{def:rank_one}
Given linear functions $w_i : V_i \rightarrow \mathbb{F} $, we define the \textbf{rank-1 tensor} associated to $(w_1,\dots,w_\ell)$ as 
   $    w_1 w_2 \dots w_\ell (v_1,\dots, v_\ell) :=$ \begin{equation*}\prod_{ i \in [\ell]}  w_i(v_i)= [v_1\otimes \dots \otimes v_\ell] (v_1,\dots,v_\ell).
    \end{equation*}
    If $V_1=\dots=V_\ell$, then 
    we can further define the $\ell^\mathrm{th}$ power of a linear form as the rank-1 tensor
    \begin{equation*}
        v^{\ell} = v \cdots v.
    \end{equation*}
\end{definition}
\begin{remark}
For simplicity, we consider the case where $V_1=\dots=V_\ell$.
\end{remark}
\subsection{Implicit Construction}
\begin{definition}\label{def:rec_bat_ten}
Given $x_1^{(i)},\dots,x_k^{(i)} \in \mathcal{L}(G_i)$, and $\prod_{i \in [\ell]} x^{(i)}_{j_i}\in \mathcal{L}(G_1 + \dots + G_\ell)$ for all $(j_1,\dots,j_\ell) \in [k]^{\ell}$, we say that $\mathcal{L}(G_1 + \dots+G_\ell)$ is $R$-recoverable if the values of $w^{\ell}_i(v_1,\dots,v_\ell)$, i.e., the diagonal elements, can be recovered from any $R$ columns of the result of the product of $C=$
\begin{equation*}
\begin{array}{r}
    \begin{bmatrix}
        w_1 \dots w_1 &  w_1 \dots w_2 & \dots & w_k \dots w_1 & \dots & w_k \dots w_k 
    \end{bmatrix}\\(v_1,\dots,v_\ell) 
    \end{array}
\end{equation*}
with {\footnotesize{
\begin{equation*} 
    \begin{bmatrix}
    x_1^{(1)}\dots x^{(\ell)}_1(P_1) & x_1^{(1)}\dots x^{(\ell)}_1(P_2)  & \dots & x_1^{(1)}\dots x^{(\ell)}_1(P_n) \\  
        x_1^{(1)}\dots x^{(\ell)}_2(P_1) & x_1^{(1)}\dots x^{(\ell)}_2(P_2)  & \dots & x_1^{(1)}\dots x^{(\ell)}_2(P_n) \\  
         \vdots & \vdots & \ddots & \vdots& \\ 
        x_k^{(1)}\dots x^{(\ell)}_k(P_1) & x_k^{(1)}\dots x^{(\ell)}_k(P_2)  & \dots & x_k^{(1)}\dots x^{(\ell)}_k(P_n) \\  
    \end{bmatrix}.
\end{equation*}}}
Here, $C$ is the result of applying all of the possible products $w_{i_1} \dots w_{i_\ell}$ to $(v_1,\dots,v_\ell)$.
\end{definition}

\begin{definition}\label{def:diag_rook_prod_ten}
 Let $G = G_1+\dots+G_\ell$ and $D = P_1+\dots+P_n$. We call $\mathcal{C}(G,D)$ a $[n,k]_q-$\textbf{tensor rook code}, if there are bases $x_1^{(i)},\dots,x_k^{(i)} \in \mathcal{L}(G_i)$ such that  \begin{equation*}
        P_k \not\in \mathrm{supp}(x_{i_1}^{(1)} \dots x_{i_\ell}^{(\ell)})  \iff i_1 = \dots= i_\ell = k .
    \end{equation*} is satisfied. 
\end{definition}
\begin{remark}
    Definition~\ref{def:diag_rook_prod_ten} seems to only encode the ``symmetric'' tensors. This is true for infinite fields but as we will see, over finite fields this models all possible functions, so that we don't have to bother encoding more general tensors; thus, for simplicity, we only consider the symmetric case. However, it is straightforward to generalize the construction implied by Definition~\ref{def:diag_rook_prod_ten} for non-symmetric tensors.
\end{remark}
\subsection{True Generality Over Finite Fields }\label{sec:true_gen}
\begin{definition}
   A tensor $T:V^\ell \rightarrow \mathbb{F} $ is \textbf{symmetric} if 
   \begin{equation*}
T(v_1,\dots,v_i,\dots,v_j,\dots,v_\ell) = T(v_1,\dots,v_j,\dots,v_i,\dots,v_\ell)
   \end{equation*}
   for all $i,j \in [l]$, and we call the \textbf{symmetric rank} the smallest number $r$ such that there exists some $w_1,\dots,w_r \in V^*$ such that 
   \begin{equation*}
       T(v_1,\dots,v_\ell) = \sum_{i \in [r]} w_i^{\ell}(v_1,\dots,v_\ell).
   \end{equation*}
   The space of all symmetric tensors on $V$ of order $\ell$ is commonly denoted as 
   $   S^\ell(V)$.
\end{definition}
Section~7.1 of \cite{landsberg2017geometry} gives that the symmetric rank of a symmetric tensor bounds its algebraic complexity in the arithmetic circuits model; in particular, it bounds the complexity in the \textit{diagonal depth-3 circuits} or \textit{depth-3 powering circuits} sometimes denoted as $\Sigma \Pi^ \ell \Sigma$ circuits. 

We now proceed to show that any function
$f : \mathbb{F}_q^t \rightarrow \mathbb{F}_q^u$
has its complexity bounded by the symmetric rank, and thus, \textbf{our model gives a scheme to perform coded distributed computing of any function over a finite field.}

\begin{folk}
    Every function
$
f : \mathbb{F}_q^t \rightarrow \mathbb{F}_q
$
over a finite field is given by a multivariate polynomial 
$
p_f \in \mathbb{F}_q[x_1,\dots,x_t].
$
In particular, any multivariate function
$
(f_1,\dots,f_b) : \mathbb{F}_q^t \rightarrow \mathbb{F}_q^u
$
is given by a polynomial map $p_f$ where $p_{f,i} \in \mathbb{F}_q[x_1,\dots,x_t]$.
\end{folk}
\begin{folk}
    Every degree $\ell$ polynomial $f$ on $t$  variables is the specialization of a symmetric tensor on $T_f \in S^\ell(\mathbb{F}_q^{t+1})$; \emph{i.e.}, 
    $
        f(x_1,\dots,x_t) = T_f(x_1,\dots,x_t,1).
        $
\end{folk}

\begin{corollary}
    Every function
$
f : \mathbb{F}_q^t \rightarrow \mathbb{F}_q^u
$
   is given by some symmetric tensor $T_f \in \left(S^\ell(\mathbb{F}_q^{t})\right)^u$.
\end{corollary}

\subsection{Explicit Construction}
We define 
\begin{equation*}
    x_i := \prod_{ j\in [k]\setminus \{i\}} z^{-1}_i.
\end{equation*}
just as before, but now define the code by sending the $\ell$ coded vectors
\begin{equation*}
    \tilde{w_j}(P_u) := \sum_{i \in [k]} w_i(v_1,\dots,v_\ell)x_i^{(j)}(P_u),  \\ 
\end{equation*}
to worker $u$, where $x^{(1)}_i = ... = x^{(\ell)}_i = x_i$, so that $\tilde w _1 =...=\tilde w_\ell =: \tilde w $.
The workers then return
$\tilde{w}^\ell(P_u)$.

\begin{theorem}\label{thm:main_result_ten}
 The construction of the $z_i$ satisfies Equation~\ref{eq:nice_resi}; in particular, the construction given by Equation~\ref{eq:main_rook} has recovery threshold given by 
 $R = \ell \sigma_k(F)$. 
\end{theorem}

\subsection{Analysis: Upper bounds on the Recovery Threshold for Tensors}
The bound in Theorem~\ref{thm:main_result_ten} can be given a coarse upper bound as follows: 
\begin{theorem}
   The codes defined by Equation~\ref{eq:main_rook} satisfy
$    R(n,k,q) \leq g_{n,q}\ell k
   $
   where $g_{n,q}$ is the smallest genus of a function field over $\mathbb{F}q$ with at least $g_{n,k}k+n$  rational places.
\end{theorem}

\begin{corollary}
    Let $\mathcal{T}(f)$ be the $\ell$-linear complexity of computing the polynomial function $f$ of degree $\ell$, then the recovery threshold of a scheme that computes the corresponding function is bounded as follows
    \begin{equation*}
        \mathcal{T} \leq R \leq  \ell \sigma_{\mathcal{T}}(F) \leq \ell g_{n,q} \mathcal{T} .
    \end{equation*}
\end{corollary}

\section*{Acknowledgment}

The first author is  partially supported by NSF DMS-2201075 and the Commonwealth Cyber Initiative.

\bibliographystyle{IEEEtran}
\bibliography{bib, skeleton}
\end{document}